\documentclass[preprint,11pt]{imsart}

\RequirePackage[numbers]{natbib}
\RequirePackage[colorlinks,citecolor=blue,urlcolor=blue]{hyperref}
\usepackage[utf8]{inputenc}
\usepackage[margin=3cm]{geometry}
\usepackage{amsmath,amssymb,amsthm,graphicx}
\usepackage{algorithm}
\usepackage{algpseudocode}
\usepackage{dsfont} 
\usepackage{booktabs}
\usepackage{caption}
\usepackage{subcaption}
\usepackage{tikz}
\usetikzlibrary{shapes,decorations,arrows,calc,arrows.meta,fit,positioning}
\usepackage{multirow} 
\usepackage{color}

\newtheorem{theorem}{Theorem}[section]

\newtheorem{proposition}[theorem]{Proposition}
\newtheorem{lemma}[theorem]{Lemma}
\newtheorem{remark}[theorem]{Remark}
\newtheorem{example}[theorem]{Example}

\newcommand{\R}{\mathbb{R}}
\newcommand{\E}{\mathbb{E}}

\begin{document}
 
\begin{frontmatter}


\title{Confidence in Causal Inference under Structure Uncertainty in Linear Causal Models with Equal Variances}

\runtitle{Confidence in Causal Inference under Structure Uncertainty}

\author{\fnms{David} \snm{Strieder}\corref{}\ead[label=e1]{david.strieder@tum.de}}
\and
\author{\fnms{Mathias} \snm{Drton}\ead[label=e2]{mathias.drton@tum.de}}
\address{Technical University of Munich; TUM School of Computation, Information and Technology, \\
	Munich Center for Machine Learning (MCML), Munich Data Science Institute (MDSI)  \\
	\printead{e1,e2}}

\runauthor{D. Strieder and M. Drton}

\begin{abstract}
 {Inferring the effect of interventions within complex systems is a fundamental problem of statistics. A widely studied approach employs structural causal models that postulate noisy functional relations among a set of interacting variables. The underlying causal structure is then naturally represented by a directed graph whose edges indicate direct causal dependencies.  In a recent line of work, additional assumptions on the causal models have been shown to render this causal graph identifiable from observational data alone.  One example is the assumption of linear causal relations with equal error variances that we will take up in this work.  When the graph structure is known, classical methods may be used for calculating estimates and confidence intervals for causal effects.  However, in many applications, expert knowledge that provides an a priori valid causal structure is not available. Lacking alternatives, a commonly used two-step approach first learns a graph and then treats the graph as known in inference.  This, however, yields confidence intervals that are overly optimistic and fail to account for the data-driven model choice.  We argue that to draw reliable conclusions, it is necessary to incorporate the remaining uncertainty about the underlying causal structure in confidence statements about causal effects.  To address this issue, we present a framework based on test inversion that allows us to give confidence regions for total causal effects that capture both sources of uncertainty: causal structure and numerical size of nonzero effects.}
\end{abstract}

\begin{keyword}[class=MSC]
\kwd[Primary ]{62D20, 62H22}
\end{keyword}

\begin{keyword}
\kwd{confidence intervals}
\kwd{causal effects}
\kwd{linear structural equation models}
\kwd{equal error variances}
\kwd{graphical models}
\end{keyword}

\end{frontmatter}

\section{Introduction}
    Questions about causal relations are at the heart of many research problems. Researchers across various fields, ranging from economics to biology and medicine, are interested in distinguishing causes from effects to predict the outcome of interventions in complex systems. While controlled experiments provide an established method for inferring causal relations, in many applied settings interventional studies are not feasible for ethical or cost considerations. Thus, inferring causal relationships based solely on observational data is an important task that the field of causal discovery addresses. 
	
    In the last decades, causal discovery has gained popularity, with much fundamental research being done on the topics of identification and estimation \citep{Pearl09,spirtes:book}. Identifiability results clarify when it is theoretically possible to go beyond mere correlations. More specifically, they clarify under which circumstances assumptions lead to unique identification of interventional distributions from a given joint probability distribution. Moreover, numerous structure learning algorithms have been proposed that infer the underlying causal structure based on observed data. Knowledge of the underlying causal structure then allows researchers to reason about the effect of interventions in complex systems and given a causal structure, standard statistical methods may be applied to estimate the involved causal effects and provide an uncertainty assessment for the numerical size of the effects. However, in the typical applied setting, the exact causal structure is (at least partially) unknown.  In order to draw reliable conclusions and make calibrated confidence statements, the remaining uncertainty about the underlying causal structure needs to be incorporated in causal inference. 

    Without prior expert knowledge about the underlying causal structure, a commonly employed ``naïve'' approach would involve splitting the task. In the first step, a causal learning algorithm estimates the underlying causal structure. In the second step, within the inferred model, confidence intervals for the causal parameters can be calculated using classical statistical inference methods. However, the resulting intervals are not valid if the causal learning algorithm infers the wrong causal structure. Such a two-step method conditions away the uncertainty in the causal structure arising from the data-driven model choice and the classical ‘double-dipping’ problem occurs since both causal structure and the effect are estimated from one data set. Thus, the approach is overly optimistic in its conclusion about the strength and existence of causal effects, and its confidence regions do not achieve the desired coverage probability, especially under high uncertainty with respect to the underlying structure.

    Despite the popularity and growing interest in causal discovery, we are unaware of prior attempts to rigorously account for uncertainty in the causal structure when providing confidence statements about causal effects. Here, we use the term confidence in the technical sense for a set with a given desired frequentist coverage probability. A completely different approach for uncertainty quantification that we do not consider here is to form Bayesian credible sets. In fact, a number of authors have used Bayesian approaches for uncertainty quantification in causal discovery, but primarily with an emphasis on the uncertainty in the graphical structure as opposed to causal effects; see \cite{10.5555/1795114.1795143},  \cite{10.5555/3020652.3020677}, and \cite{cao:2019} for three selected examples. 
    
    In this article, we want to shed light on the problem by highlighting the challenges of the task of handling structure uncertainty and by proposing a solution for constructing confidence intervals for total causal effects that capture both types of uncertainty, the uncertainty regarding the causal structure resulting from a data-driven model selection and the uncertainty about the numerical size of the effect.
    Our solution is feasible for problems of small to moderate scale---e.g., our simulation experiments will  consider problems with up to 12 variables.  However, already at this scale, rigorously accounting for structure uncertainty means to account for more than $10^{26}$ possible structures, reflecting the fast superexponential growth with the number of nodes in the causal graph.
    
    
    In Section \ref{section:background} of the paper, we review the considered framework of structural causal models \citep{peters:book, handbook}. More specifically, in order to ensure identifiability from observational data alone, we consider a restricted class of causal models, namely linear structural equation models with Gaussian errors and equal variances \citep{PetersEV} and review the definition of a total causal effect, which is the target of interest in this study. We then generalize and improve the ideas introduced for bivariate problems in \citet{UAI21:Strieder} and present a general framework for constructing confidence intervals for total causal effects in multivariate data (Section \ref{section:testinversion} and \ref{section:testing}). Furthermore, we present computational shortcuts and  analyze the performance of the introduced framework in numerical experiments (Section \ref{section:computation}). In the concluding Section \ref{section:discussion} we discuss extensions of the framework to partially quantify the uncertainty over causal structures.
    
\section{Background}\label{section:background}

    This section reviews linear structural equation models and the total causal effect which is the target of interest in this study.  
	
\subsection{Linear Structural Equation Models}
			
    Structural causal models constitute a commonly employed tool to study causal relationships and represent noisy causal relations among a set of interacting variables. Each variable is modelled as a function of a subset of other variables, its causes, and a stochastic error term. The causal perspective emerges from viewing those relations as making assignments rather than representing mathematical equations. The left-hand side, the effect, is assigned the value specified on the right-hand side, the causes. Externally varying the values on the right-hand side results in a change on the left-hand side, but not vice versa. This reflects the asymmetry inherent in cause-effect relations. 
    
    We assume access to observational data in the form of a sample of independent copies of a random vector $X=(X_1, \dots , X_d)$; without loss of generality, we assume the random vector to have zero mean. Without any further assumptions on the structural causal model, e.g. the involved functions, identification of the underlying causal structure is only possible up to a Markov equivalence class.  Therefore, to ensure unique identifiability, following a line of research initiated by \cite{PetersEV}, we focus on linear relations and normal distributed errors with equal variances by assuming that  $X$ solves the equation system 
	\begin{equation} \label{eq:LSEM}
    X_j=\sum_{i\neq j}\beta_{j,i}X_i + \varepsilon_j, \quad j=1, \dots ,d.
    \end{equation}
    Here, $B:=[\beta_{j,i}]_{j,i=1}^d$ are unknown parameters that represent the direct causal effects between the variables, and $\varepsilon_j$ are independent, normally distributed error terms
    \begin{equation*}
                \varepsilon_j \overset{i.i.d.}{\sim} \mathcal{N}(0,\sigma^2), \quad j=1, \dots ,d,
    \end{equation*} 
    with common, unknown variance $\sigma^2>0$. Each specific \textit{Linear Structural Equation Model} (LSEM) restricts a subset of the direct effects $\beta_{j,i}$ to be zero. Thus, LSEMs are naturally represented by a (minimal) directed graph $G(B)$ with vertex set $V=\{1, \dots , d\}$. In the associated directed graph, the edge set equals the support of $B$, thus, a missing edge $i \rightarrow j$ indicates that $\beta_{j,i}=0$. As in related work, we assume the underlying directed graph to be acyclic (DAG), which entails that $B$ is permutation similar to a triangular matrix and the system \eqref{eq:LSEM} admits the unique solution $X=(I_d-B)^{-1} \varepsilon$, with $I_d$ denoting the $d\times d$ identity matrix. Incorporating the equal variance assumption, the covariance matrix of $X$ is seen to be 
    \begin{equation*} \E[XX^{T}]=\sigma^2(I_d-B)^{-1}(I_d-B)^{-T}.
    \end{equation*}
    
    In the sequel, we will use the following graphical concepts. Each DAG admits a (not necessarily unique) topological ordering of its vertices, that is, a permutation $\pi$ of $\{1, \dots , d\}$ such that the existence of a directed path from node $\pi(i)$ to node $\pi(j)$ implies $i<j$. We write $i<_{G}j$ if node $i$ precedes node $j$ in such a causal ordering of the graph. If the DAG contains an edge from node $i$ to node $j$, then node $i$ is a parent of node $j$, and we denote the set of all parents of node $j$ with $p(j)$. It has been shown that under the above modelling assumptions with equal error variances, i.e., homoscedasticity across all interacting variables, the causal ordering is identifiable and corresponds to an ordering of conditional variances \citep{ChenEV,Ghoshal2018}.  Hence, causal effects can also be uniquely identified.

\subsection{Causal Effects}

    Informally put, we are interested in estimating how much the value of $X_j$ changes if we externally intervene in the system and set the value of $X_i$ to $x_i$. Such an intervention represents a change in the true data-generating process and, thus, a modification in our probabilistic model. In the structural equation framework, the effect of such an external intervention can be expressed by replacing the $i$-th equation with $X_i=x_i$ and is denoted as $do(X_i=x_i)$. 
	
    In this work, we seek to construct confidence intervals for the total causal effect $\mathcal{C}(i\rightarrow j)$ in LSEMs. This effect is formally defined as the unit change in the expectation of $X_j$ with respect to an intervention in $X_i$, that is,
    \begin{equation*}
    \mathcal{C}(i\rightarrow j):=\frac{\text{d}}{\text{d} x_i} \E[X_j|\text{ do}(X_i=x_i)].   
    \end{equation*}

    \begin{remark}\label{remark:paths}
    Note that the total causal effect in an LSEM $(B,\sigma^2)$ can be intuitively expressed based on the respective underlying DAG $G(B)$. The effect  $\mathcal{C}(i\rightarrow j)$ equals the sum of all products of edge weights along all directed paths from $i$ to $j$. Furthermore, the effect is zero if there does not exist a directed path from $i$ to $j$. This follows from the path properties of the matrix $(I_d-B)^{-1}$ that computes the solution $X$ of the structural equations from the independent errors. 
    \end{remark}

    \begin{example}
    Consider the following LSEM and the corresponding (minimal) DAG:
    
    \begin{minipage}[b]{0.45\linewidth}
    \begin{align*}
        X_1&=\beta_{1,3} X_3 + \varepsilon_1 \\
        X_2&=\beta_{2,1} X_1 + \beta_{2,4} X_4 + \beta_{2,5} X_5 + \varepsilon_2\\
        X_3&=\varepsilon_3 \\
        X_4&=\beta_{4,1} X_1 + \varepsilon_4 \\
        X_5&=\beta_{5,3} X_3 + \varepsilon_5 \\
    \end{align*}
    \end{minipage}
    \begin{minipage}[b]{0.5\linewidth}
        \tikzset{
          every node/.style={circle, inner sep=1mm, minimum size=0.8cm, draw, thick, black, fill=white, text=black},
          every path/.style={thick}
        }      \scalebox{.8}{
        \begin{tikzpicture}[->,>=triangle 45,shorten >=1pt,
          auto,thick, main
          node/.style={circle,fill=gray!20,draw,font=\sffamily\bfseries}]

            \node[main node] (1) at (2,0)     {$1$};
            \node[main node] (2) at (5,-4)    {$2$};
            \node[main node] (3) at (0,-2)     {$3$};
            \node[main node] (4) at (5,0)     {$4$};
            \node[main node] (5) at (2,-4)    {$5$};

            \path[color=red,every node/.style={font=\sffamily\small}] 
            (1) edge[line width=1mm] node {$\beta_{4,1}$} (4) (4) edge[line width=1mm] node {$\beta_{2,4}$} (2);
            \path[color=red,every node/.style={font=\sffamily\small}]
            (1) edge[line width=1mm] node {$\beta_{2,1}$} (2);
            \path[color=black,every node/.style={font=\sffamily\small}]
            (3) edge[line width=1mm] node {$\beta_{5,3}$} (5) (5) edge[line width=1mm] node {$\beta_{2,5}$} (2);
            \path[color=black,every node/.style={font=\sffamily\small}] 
            (3) edge[line width=1mm] node {$\beta_{1,3}$} (1);
        \end{tikzpicture} }
    \end{minipage}
    The total causal effect is given by the sum of all products of edge weights along all directed paths. For example, we obtain $\mathcal{C}(1\rightarrow 2)=\beta_{2,1}+\beta
_{4,1}\beta_{2,4}$, whereas $\mathcal{C}(2\rightarrow 1)=0$.
    \end{example}

    Our aim is to construct a confidence interval for the total causal effect $\mathcal{C}(i\rightarrow j)$ when the underlying causal structure is not known and has to be learned. We stress that, under the identifiability assumption of an underlying Gaussian LSEM with equal error variances, this is a well-defined problem. In fact, every admissible distribution can be represented by a single (minimal) DAG. This DAG then uniquely defines the total causal effect.
    
   A ``naïve'' two-step approach that first learns a causal structure and then calculates confidence intervals in the inferred model does not incorporate the uncertainty regarding the underlying causal structure. To respect this uncertainty in the calculations of a valid confidence interval, it is necessary to take all possible structures, i.e. graphs, into account. As a result, the target quantity has a composite structure and, thus, obtaining finite sample distributions of estimators is difficult. 
   
   In many situations where the sampling distribution of estimators is unobtainable, bootstrapping and other resampling methods can offer an appealing framework to construct confidence intervals.  At first sight, resampling may appear to provide a simple solution to the problem under discussion. This solution relies on computing for each resampled data set a causal effect estimate. To obtain this estimate one may compose a consistent model selection method that learns the causal graph with a consistent causal effect estimator in the learned model. Note that this composition is well-defined due to the identifiability of LSEMs with equal error variances. However, as a statistic, this composition lacks smoothness and, thus, bootstrapping procedures may fail severely \citep{andrews:guggenberger:2010,drton:williams:2011}. We emphasize that \citet{drton:williams:2011} demonstrate that even in low dimensions, the asymptotic behavior of confidence intervals and bootstrap tests is difficult to predict for complex composite settings as encountered here. The subtleties stem from the fact that the set of covariance matrices associated with at least one possible DAG form a union of smooth manifolds. This union contains singularities at the intersections of the manifolds, which is highlighted in the bivariate case  by \citet{UAI21:Strieder}. They demonstrate that singular points arise at the intersections of the two models. Consequently, bootstrapping cannot be used to capture model uncertainty correctly. This failure of the bootstrapping methods can also be observed in our simulations in Section \ref{section:simulation}.

    In conclusion, there is a need to develop new frameworks for constructing confidence intervals that circumvent the problems that originate from the non-smooth nature of the target of interest. 

\section{Confidence Intervals via Test Inversion}\label{section:testinversion}
     In the remainder of this paper, we develop a framework that circumvents the problems posed by the non-smooth nature of the total causal effect and provides valid confidence intervals for causal inference under structure uncertainty. In this section, we introduce our main ansatz, the test inversion approach, and derive the concrete hypothesis that subsequently needs to be tested.

\subsection{Inversion of Tests}
	
    Our main idea is to leverage the classical duality between statistical hypothesis tests and confidence regions in the following way. If we perform valid hypothesis tests for all attainable causal effects, then all values that we cannot reject form a confidence region for the total causal effect. More specifically, let $\alpha \in (0,1)$ be a fixed significance level. We perform (asymptotically) valid level $\alpha$ tests for the hypothesis of a fixed total causal effect $\psi$ for all attainable values of the causal effect, that is, for all $\psi \in \R$, we test 
	\begin{align*}
	    \text{H}_0: \mathcal{C}(1 \rightarrow 2) = \psi
	\end{align*}
	with a test whose acceptance region we denoted by $A(\psi)$. Then a (asymptotically) valid $(1-\alpha)$-confidence region for the total causal effect $\mathcal{C}(1 \rightarrow 2)$ for the data $X$ is given by 
	\begin{align*}
	    C(X):=\{\psi \in \R: X \in A(\psi) \}.
	\end{align*}
    A proof of this classical result can be found, e.g., in  \cite[Theorem 9.2.2]{Casella}.

    With the above perspective, our problem is shifted to developing suitable hypothesis tests for all attainable causal effects. Due to the lack of concrete knowledge about the model, all possible causal graphs must be respected. Thus, we have to construct hypothesis tests in a composite setting that remains challenging. In the next section, we explicitly describe the hypotheses that need to be tested to construct confidence regions for the total causal effect.

\subsection{Hypothesis of a Fixed Causal Effect}
	
    Recalling the assumptions of an underlying LSEM with homoscedastic errors across all interacting variables, we consider the following task. We assume that the given data consists of random vectors $X^{(1)}, \dots , X^{(n)}$ drawn independently from a centered multivariate normal distribution $N(0,\Sigma_0)$.  The distribution's density $p(x|\Sigma_0)$ factorizes according to an unknown DAG while satisfying the assumption of equal error variances. To account for the uncertainty in the causal structure, we need to test at level $\alpha$ whether any multivariate centered normal distribution that factorizes to any DAG under the equal variance assumption allows for a total causal effect of size $\psi$, and repeat that test for all values $\psi \in \R$. In light of the normality assumption, it is natural to parameterize our statistical model in terms of the covariance matrix $\Sigma$. Further, the additional assumption of homoscedasticity across all interacting variables restricts the set of possible covariance matrices further as we detail now.
    
    Given a fixed DAG $G$, under the assumption of equal error variances, all possible normal distributions with a density that factorizes according to $G$ are defined by the family  
	        \begin{equation}\label{models:G}
	            \mathcal{N}(G):=\{ N(0,\Sigma): \exists B \in \R^G, \sigma^2 >0 \text{ with } \Sigma = \sigma^2(I_d-B)^{-1}(I_d-B)^{-T} \} ,
            \end{equation}
	        where $\R^G=\{B \in \R^{d \times d} : \beta_{j,i}=0 \text{ if } i \notin p(j) \}$.
	Considering and combining all possible DAG models, our overall model corresponds to the following union of sets of covariance matrices:
	        \begin{equation}\label{definition:model}
                \mathcal{M}:=\bigcup_{G \in \mathcal{G}(d)} \{\Sigma \in \text{PD}(d) : N(0,\Sigma) \in \mathcal{N}(G) \},
            \end{equation}
    where $\mathcal{G}(d)$ is the set of all DAGs with $d$ nodes.

    \begin{remark}\label{remark:ordering}
        If $G^*$ is a supergraph of $G$, then the set $\mathcal{N}(G)$ is a subset of $\mathcal{N}(G^*)$.  Hence, in the definition of $\mathcal{M}$, we only need to consider the union over all complete DAGs with $d$ nodes, i.e. all possible orderings.
    \end{remark}
    
    The following Proposition shows that we can identify the model $\mathcal{M}$, which is a subset of all positive definite matrices satisfying the assumption of equal error variances, solely based on polynomial constraints in the entries of the covariance matrix. Here, and in the remainder of this article, we write $\Sigma_{i,j|p(i)}$ for the conditional covariance matrix, that is,
    \begin{equation*}
        \Sigma_{j,i|p(i)}:=\Sigma_{j,i}-\Sigma_{j,p(i)}(\Sigma_{p(i),p(i)})^{-1}\Sigma_{p(i),i}.
    \end{equation*}
    \begin{proposition}\label{theorem:thm1} For a covariance matrix $\Sigma \in$ PD$(d)$, the following statements are equivalent:
    \begin{enumerate}
        \item[(i)] $ \Sigma \in \mathcal{M}$, 
        \item[(ii)] There exists a complete DAG $G$ and $\sigma^2 >0$ such that for all $k=1, \dots ,d$ 
    \end{enumerate}  
    \begin{equation}\label{eq:thm1}
        \sigma^2= \Sigma_{k,k|p(k)}=\Sigma_{k,k}-\Sigma_{k,p(k)}(\Sigma_{p(k),p(k)})^{-1}\Sigma_{p(k),k}.
    \end{equation} 
    \end{proposition}
    \begin{proof}
        ``$\implies$'' Let $\Sigma \in \mathcal{M}$. Then there exists a DAG $G$ such that $N(0,\Sigma) \in \mathcal{N}(G)$. Considering Remark \ref{remark:ordering}, we can assume without loss of generality that $G$ is complete. Straightforward calculations with the multivariate normal density then imply that the variances of each node conditioned on its parents are equal to $\sigma^2$; or see a combinatorial argument as in \cite[Equation (7.4)]{drton:2018}. Thus, for all $k=1, \dots , d$,
        \begin{equation*}
            \sigma^2=\text{Var}(X_k|X_{p(k)})=\Sigma_{k,k}-\Sigma_{k,p(k)}(\Sigma_{p(k),p(k)})^{-1}\Sigma_{p(k),k}.
        \end{equation*}
        
        ``$\impliedby$'' Since $G$ is complete, the multivariate normal distribution $N(0,\Sigma)$ factorizes according to the graph $G$. This implies the existence of a matrix $B \in \R^G$ and a positive definite diagonal matrix $\Omega \in \R^{d \times d}$ with $\Sigma = (I_d-B)^{-1}\Omega(I_d-B)^{-T}$. Equation \eqref{eq:thm1} implies that $\Omega_{k,k}=\Sigma_{k,k|p(k)}$ is equal for all $k=1, \dots , d$ and the claim follows.
    \end{proof}

    \begin{remark}\label{remark:unique}
        Note the identifiability result under homoscedasticity across all interacting variables, e.g., Theorem 1 in \cite{ChenEV}. Each $\Sigma \in \mathcal{M}$ corresponds to a unique $(B,\sigma^2)$ and, thus, a unique minimal DAG $G(B)$. As explained in Remark \ref{remark:ordering}, the union in the definition of $\mathcal{M}$ \eqref{definition:model} is not necessarily disjoint and analogously the (complete) graph satisfying Condition 2 in Proposition \ref{theorem:thm1} is not necessarily unique. Nevertheless, since $B$ is unique, all graphs satisfying Condition 2 in Proposition \ref{theorem:thm1} share the same edge set with non-zero edge coefficients and correspond to all valid causal orderings of the minimal DAG. Thus, every 'additional' parent of node $i$ in a complete DAG $G$ from Condition 2 compared to the unique minimal DAG $G(B)$ is conditionally independent of node $i$ given $p_{G(B)}(i)$. This follows by d-separation since every (undirected) path between node $i$ and an additional parent either traverses the parents $p_{G(B)}(i)$ in the minimal DAG or contains a collider $c$ with $i <_{G(B)} c$ and thus this collider has no descendent in $p_{G(B)}(i)$.
    \end{remark}
 
    Our testing problem is the following classical statistical problem. We are concerned with a parametric model $\mathcal{M}$, parameterized by the covariance matrix, which is an intricate union of subsets of the cone of positive definite matrices given by polynomial constraints. Within that statistical model, we need to test for all attainable values of the total causal effect $\psi$. In the following Proposition, we show that the parameter of interest is a function of the covariance matrix. Thus, the considered statistical testing problem can be fully parameterized in terms of the covariance matrix.
	
    \begin{proposition}\label{theorem:effect}
    Let $\Sigma \in \mathcal{M}$. Then the total causal effect $\mathcal{C}(i \rightarrow j)$ is given by
	\begin{equation*} 
	    \mathcal{C}(i \rightarrow j)=\Sigma_{j,i|p(i)}(\Sigma_{i,i|p(i)})^{-1},
	\end{equation*}
    with parent set $p(i)$ based on any corresponding DAG $G$ from Proposition \ref{theorem:thm1}.
    \end{proposition}
    \begin{proof}
        The case $j <_G i$ is trivially true, since in this case, the total causal effect $\mathcal{C}(i \rightarrow j)$ is zero and, considering $G$ is complete, the parent set $p(i)$ contains the node $j$, which yields $\Sigma_{j,i|p(i)}=0$. 

        Assume $i <_G j$. Since the set $p(i)$ satisfies the back-door criterion with respect to $i$ and $j$, we have 
        \begin{equation*}
            \mathcal{C}(i\rightarrow j):=\frac{\text{d}}{\text{d} x_i} \E[X_j|\text{ do}(X_i=x_i)]=\frac{\text{d}}{\text{d} x_i} \E\big[\E[X_j|X_i=x_i,X_{p(i)}]\big].
        \end{equation*}
        It follows that the total causal effect $\mathcal{C}(i\rightarrow j)$ is the regression coefficient of $X_i$ when regressing $X_j$ on $(X_i, X_{p(i)})$, which is given by $\Sigma_{j,i|p(i)}(\Sigma_{i,i|p(i)})^{-1}$.
        
        We emphasize that this expression does not depend on the specific choice of the DAG $G$. All graphs satisfying Condition 2 in Proposition \ref{theorem:thm1} share the same egde set with non-zero edge weights and thus the parent set $p(i)$ always satisfies the back-door criterion in the corresponding unique minimal DAG, see Remark \ref{remark:unique}.
    \end{proof}
    Using Proposition \ref{theorem:effect}, we can define the hypothesis space $\mathcal{M}_\psi \subset \mathcal{M}$, given by all covariance matrices $\Sigma \in \mathcal{M}$ that allow for a total causal effect of size $\psi$ via
    
    \begin{align*}
	    \mathcal{M}_\psi:=\Big\{\Sigma \in \text{PD}(d): \exists \text{ complete DAG } G \text{ and } \sigma^2 >0 \text{ with } 
	    \psi \sigma^2&=\Sigma_{j,i|p(i)} \\ \text{ and }
	    \sigma^2&=\Sigma_{k,k|p(k)} \quad  \forall\  k=1, \dots , d
	       \Big\}.
    \end{align*}
    Put together we have to solve the statistical testing problem 
    \begin{align}\label{eq:testproblem}
	    H_0^{(\psi)} :
	    \Sigma \in \mathcal{M}_\psi \quad \text{against} \quad H_1 : \Sigma \in \mathcal{M} \backslash \mathcal{M}_\psi
    \end{align}
    for all attainable $\psi \in \R$. Then we invert the tests to construct valid confidence intervals for the total causal effect. Similar to the union $\mathcal{M}$, we can write the hypothesis space as a union of single hypotheses over all (complete) DAGs, that is $H^{(\psi)}_0 := \bigcup_{G \in \mathcal{G}(d)} H^{(\psi)}_0(G)$ with single hypotheses $H^{(\psi)}_0(G): \Sigma \in \mathcal{M}_{\psi}(G)$, where
    \begin{equation}\label{eq:hypothesis}
	     \mathcal{M}_{\psi}(G) := \Big\{ \Sigma \in \text{PD}(d) : \exists \sigma^2>0 \text{ with } 
	    \psi \sigma^2=\Sigma_{j,i|p(i)} \text{ and }
	    \sigma^2=\Sigma_{k,k|p(k)} \quad  \forall\  k=1, \dots , d
	    \Big\}.
	\end{equation}

    \begin{remark}\label{remark:zeroeffect}
        It is clear that for $\psi \neq 0$ we only need to consider the union over all graphs with $i <_G j$. In contrast, for a zero effect we have to consider all graphs, since the effects of multiple directed paths might 'cancel' each other. Further, for graphs with $j<_G i$ the constraint given by the size zero effect is not informative, since $j \in p(i)$ and consequently $\Sigma_{j,i|p(i)}=0$. Therefore, the corresponding hypothesis space has higher dimension. 
    \end{remark}

    \begin{example}
         Consider the following LSEM and the corresponding (minimal) DAG:
    
    \begin{minipage}[b]{0.45\linewidth}
    \begin{align*}
        X_1&=-0.5 X_3 + \varepsilon_1 \\
        X_2&=0.25 X_1 + 0.5 X_4 + 0.25 X_5 + \varepsilon_2\\
        X_3&=\varepsilon_3 \\
        X_4&=-0.5 X_1 + \varepsilon_4 \\
        X_5&=0.5 X_3 + \varepsilon_5 \\
    \end{align*}
    \end{minipage}
    \begin{minipage}[b]{0.5\linewidth}
        \tikzset{
          every node/.style={circle, inner sep=1mm, minimum size=0.8cm, draw, thick, black, fill=white, text=black},
          every path/.style={thick}
        }      \scalebox{.8}{
        \begin{tikzpicture}[->,>=triangle 45,shorten >=1pt,
          auto,thick, main
          node/.style={circle,fill=gray!20,draw,font=\sffamily\bfseries}]

            \node[main node] (1) at (2,0)     {$1$};
            \node[main node] (2) at (5,-4)    {$2$};
            \node[main node] (3) at (0,-2)     {$3$};
            \node[main node] (4) at (5,0)     {$4$};
            \node[main node] (5) at (2,-4)    {$5$};

            \path[color=red,every node/.style={font=\sffamily\small}] 
            (1) edge[line width=1mm] node {$-0.5$} (4) (4) edge[line width=1mm] node {$0.5$} (2);
            \path[color=red,every node/.style={font=\sffamily\small}]
            (1) edge[line width=1mm] node {$0.25$} (2);
            \path[color=black,every node/.style={font=\sffamily\small}]
            (3) edge[line width=1mm] node {$0.5$} (5) (5) edge[line width=1mm] node {$0.25$} (2);
            \path[color=black,every node/.style={font=\sffamily\small}] 
            (3) edge[line width=1mm] node {$-0.5$} (1);
        \end{tikzpicture} }
    \end{minipage}
    There exist two directed paths from node $1$ to node $2$ in the corresponding (minimal) DAG and thus, $1 <_G 2$ in all valid causal orderings. However, the effect of the two path 'cancel' each other, such that $\mathcal{C}(1\rightarrow 2)=0.25-0.5^2=0$.
    \end{example}

\section{Testing for Total Causal Effects}\label{section:testing}
    For constructing confidence intervals for the total causal effect, we invert tests of the hypothesis \eqref{eq:testproblem}. In this section, we present two concrete tests based on likelihood ratio theory: (i) a classical constrained likelihood ratio test \citep{Constrained}, and (ii) an approach based on the recently proposed split likelihood ratio test \citep{Universal}. 
    
    \subsection{Constrained Likelihood Ratio Test}

    The first option we consider is to form a classical likelihood ratio statistic for the testing problem \eqref{eq:testproblem}.  Writing $\ell_n$ for the Gaussian log-likelihood function based on the sample $X^{(1)},\dots,X^{(n)}$, the likelihood ratio statistic for \eqref{eq:testproblem} is
        \begin{equation}
        \label{eq:actual_lrstat}
        \check{\lambda}^{(\psi)}_n:=2\Big(\sup_{\Sigma \in \mathcal{M}} \ell_n(\Sigma)- \sup_{\Sigma \in  \mathcal{M}_{\psi}} \ell_n(\Sigma)\Big).
    \end{equation} 
    In regular settings the statistic may be compared to a chi-square quantile for an asymptotic test.  Here, however, asymptotic distribution theory is difficult because both the null hypothesis and the alternative are intricate unions and do not satisfy the usual regularity assumptions.  In the simplified bivariate setting, \citet{UAI21:Strieder} avoid the difficulties posed by the existence of singularities by testing modified problems in two different approaches. In the first approach, the hypothesis and alternative are relaxed to test an ordering of (conditional) variances.  However, it is not feasible to optimize this modified testing problem in problems beyond the bivariate case. Further, the mixture weights of the arising limit distribution, which is a mixture of chi-square distributions, are difficult to determine in higher dimensions. The second approach relaxes the alternative to an unconstrained Gaussian model.  In this modification, the statistic \eqref{eq:actual_lrstat} is changed to a larger statistic for which the optimization of $\mathcal{M}$ is replaced by an optimization over the positive definite cone $\text{PD}(d)$. This approach, however, may cause problems in the test inversion approach, as the union of tested null hypotheses no longer coincides with the alternative.  In particular, small confidence regions may arise due to model misspecification.


    In the present paper we will derive our confidence region for the total causal effect from tests that reject for too large values of the original statistic $\check{\lambda}^{(\psi)}_n$ from \eqref{eq:actual_lrstat}.  However, we will consider relaxing the alternative $\mathcal{M}$ and employ the theory of intersection union tests to obtain a simple yet effective upper bound on the distribution of $\check{\lambda}^{(\psi)}_n$.  

    To obtain this upper bound, first note that the hypothesis we need to test is an intricate union of single hypotheses given by $ \bigcup_{G \in \mathcal{G}(d)} \mathcal{M}_{\psi}(G)$. While we will later show that every single hypothesis defines a smooth submanifold, their union is not necessarily a smooth submanifold again. Nevertheless, the original test statistic \eqref{eq:actual_lrstat} can then be written as $\check{\lambda}^{(\psi)}_n=\min_{ G \in \mathcal{G}(d)}\check{\lambda}^{(\psi)}_n(G)$, with 
    
    \begin{equation}
        \label{eq:actual_lrstatG}
        \check{\lambda}^{(\psi)}_n(G):=2\Big(\sup_{\Sigma \in \mathcal{M}} \ell_n(\Sigma)- \sup_{\Sigma \in  \mathcal{M}_{\psi}(G)} \ell_n(\Sigma)\Big).
    \end{equation} 

    Thus, the original test statistic can be upper bounded by every single test statistic $\check{\lambda}^{(\psi)}_n(G)$ and a valid level $\alpha$ test of the union of single hypotheses can be constructed by testing every single hypothesis. We then reject the union null hypothesis if we reject all single hypotheses.  A formal proof of this classical theory of intersection union tests can be found, e.g., in \cite[Theorem 8.3.23]{Casella}. 
    
    In the second step, we relax the alternative to an unconstrained Gaussian model to obtain the following larger test statistic
    \begin{equation*}
        \lambda^{(\psi)}_n(G):=2\Big(\sup_{\Sigma \in \text{PD}(d)} \ell_n(\Sigma)- \sup_{\Sigma \in  \mathcal{M}_{\psi}(G)} \ell_n(\Sigma)\Big) 
    \end{equation*}
    
    Relaxing the alternative to an optimization over the entire cone of positive definite matrices PD$(d)$ potentially increases the achieved likelihood and consequently upper bounds each single test statistic $\check{\lambda}^{(\psi)}_n(G)$.
    
    Combining both approximations yields an upper bound of the original test statistic $\check{\lambda}^{(\psi)}_n$ and thus, a simple way to construct asymptotically conservative hypothesis tests. In the following, we state our main result by turning them into conservative confidence intervals for the total causal effect $\mathcal{C}(i \rightarrow j)$ via test inversion. 
	\begin{theorem}\label{theorem:lrtinterval}
	    Let $\alpha \in (0,1)$. Then an asymptotic $(1-\alpha)$-confidence set for the causal effect $\mathcal{C}(i\rightarrow j)$ is given by
        \begin{align*}
            C = \{\psi \in \R  :  \check{\lambda}_n^{(\psi)}(i <_G j)\leq \chi_{d,1-\alpha}^2 \} \cup \{0 : \check{\lambda}_n^{(0)}(j <_G i)\leq \chi_{d-1,1-\alpha}^2\},
        \end{align*}
        where $\check{\lambda}_n^{(\psi)}(i <_G j):=\min_{ G \in \mathcal{G}(d)\,:\, i <_G j} \check{\lambda}_n^{(\psi)}(G)$ and $\check{\lambda}_n^{(0)}(j <_G i):=\min_{ G \in \mathcal{G}(d)\,:\, j <_G i}\check{\lambda}_n^{(0)}(G)$.
	\end{theorem}
	
	\begin{remark}
	   Note that we only need to consider complete DAGs, which drastically reduces the number of single hypothesis tests that need to be evaluated and thus the computational burden. Further, we need to respect graphs with $j <_G i$ only for zero-sized effects. Therefore, our proposed confidence regions may contain a non-zero interval and an isolated zero. In contrast to the bivariate case, however, zero effects can also correspond to orderings with $i <_G j$ due to cancellation, see also Remark \ref{remark:zeroeffect}.
	\end{remark}

    \begin{proof}
    Let $\psi \in \R$ and $\Sigma \in \mathcal{M}_{\psi}$. Since $\mathcal{M}_{\psi}=\bigcup_{G \in \mathcal{G}(d)} \mathcal{M}_{\psi}(G)$ there exists a complete graph $G$, such that  $\Sigma \in \mathcal{M}_{\psi}(G)$. If $i <_G j$, it follows from the introduced upper bound and Lemma \ref{theorem:limit1}
    \begin{equation*}
        P_{\Sigma}\Big( \check{\lambda}^{(\psi)}_n > \chi_{d,1-\alpha}^2 \Big) \leq   P_{\Sigma}\Big( \check{\lambda}^{(\psi)}_n(G) > \chi_{d,1-\alpha}^2 \Big)   \leq P_{\Sigma}\Big( \lambda^{(\psi)}_n (G) > \chi_{d,1-\alpha}^2 \Big) \rightarrow \alpha.
    \end{equation*} 
    We obtain an analogous result for $j <_G i$ with Lemma \ref{theorem:limit2}. Employing the test inversion approach yields the claim.
    \end{proof}


    In the following, we derive the asymptotic distribution of the upper bound $\lambda^{(\psi)}_n(G)$ of our test statistic under every single hypothesis $\mathcal{M}_{\psi}(G)$. As previously indicated, the cases $i <_G j$ versus $j <_G i$ define different dimensional hypothesis spaces, thus we consider both instances separately. 
   
    First, we consider the case $i <_G j $ and assume without loss of generality that $(1,2, \dots, d)$ is the causal ordering. In this case, the single hypothesis can be written as $\mathcal{M}_{\psi}(G) = \{ \Sigma \in \text{PD}(d) : f_\psi(\Sigma|G)=0\}$, with
    \begin{equation*}
	    f_\psi(\Sigma|G):=\begin{pmatrix}
	    \Sigma_{j,i|p(i)}-\psi \Sigma_{1,1} \\ 
	    \Sigma_{2,2|p(2)}-\Sigma_{1,1} \\
	    \vdots \\
	    \Sigma_{d,d|p(d)}-\Sigma_{1,1}
	    \end{pmatrix}.
    \end{equation*}
    The relaxed model space of all positive definite matrices PD$(d)$ can be identified with an open subspace in $\R^{\frac{1}{2}(d^2+d)}$. The hypothesis space $ \mathcal{M}_{\psi}(G)$ then defines a $\frac{1}{2}(d^2-d)$-dimensional submanifold of $\R^{\frac{1}{2}(d^2+d)}$, since the Jacobian of $f_\psi$ has full rank $d$. To see that, note its (nonzero) triangular structure for derivatives for $(\Sigma_{k,k})_{k=2,\dots,d}$ in all rows except the first row. The first row, however, has a nonzero derivative for $\Sigma_{j,i}$, and zero derivatives for $(\Sigma_{k,k})_{k=j,\dots,d}$. 
	
    Using the fact that this single hypothesis defines a smooth submanifold, we determine the limit distribution of the upper bound $\lambda^{(\psi)}_n(G)$.
	
	\begin{lemma}\label{theorem:limit1}
	    Let $G \in \mathcal{G}(d)$ be a DAG with $i <_G j$ and $\psi \in \R$. Under the hypothesis $H_0^{(\psi)}(G)$ the likelihood ratio test statistic $\lambda^{\psi}_n(G)$ satisfies 
	    \begin{equation*}
	        \lambda^{(\psi)}_n(G) \overset{\mathcal{D}}{\rightarrow} \chi^2_d.
	    \end{equation*}
	\end{lemma}
	\begin{proof}
	   Since the hypothesis space $ \mathcal{M}_{\psi}(G)$ defines a $\frac{1}{2}(d^2-d)$-dimensional submanifold of $\R^{\frac{1}{2}(d^2+d)}$, the tangent cone at every $\Sigma \in  \mathcal{M}_{\psi}(G)$ is a $\frac{1}{2}(d^2-d)$-dimensional linear subspace of $\R^{\frac{1}{2}(d^2+d)}$ and the claim follows. For details on this classical result, we refer to \cite{drton:2009}.
	\end{proof}

        In the case $j <_G i$ the causal effect from $i$ to $j$ can only be zero. Therefore, we solely need to test the single hypothesis $H^{(0)}_0(G)$. We assume again without loss of generality that $(1, \dots , d)$ is the underlying causal order. The polynomial constraint in the hypothesis \eqref{eq:hypothesis} corresponding to the fixed size effect is not informative, as the node $j$ is included in $p(i)$ and thus $\Sigma_{j,i|p(i)} = 0$. Therefore, the hypothesis $H^{(0)}_0(G)$ is solely defined by $d-1$ polynomial constraints corresponding to the assumption of equal error variances, that is, in this case we have $ \mathcal{M}_{0}(G)= \{ \Sigma \in \text{PD}(d) : f_0(\Sigma|G)=0\}$, where \begin{equation*}
	    f_0(\Sigma|G):=\begin{pmatrix}
	       \Sigma_{2,2|p(2)}-\Sigma_{1,1} \\
	    \vdots \\
	    \Sigma_{d,d|p(d)}-\Sigma_{1,1}
	    \end{pmatrix}.
	\end{equation*}
	
    Thus, similarly to the previous case, we obtain a chi-square distribution as the limit distribution, however with fewer degrees of freedom.
       \begin{lemma}\label{theorem:limit2}
	    Let $G \in \mathcal{G}(d)$ be a DAG with $j <_G i$. Under the hypothesis  $H_0^{(0)}(G)$ the likelihood ratio test statistic $\lambda^{(0)}_n(G)$ satisfies 
	    \begin{equation*}
	        \lambda^{(0)}_n(G) \overset{\mathcal{D}}{\rightarrow} \chi^2_{d-1}.
	    \end{equation*}
	\end{lemma}
	\begin{proof}
	    Analogously to Lemma \ref{theorem:limit1}. Note that the Jacobian has full rank $d-1$ due to its triangular structure.
	\end{proof}
	
	With Theorem \ref{theorem:lrtinterval} we defined an asymptotically valid confidence interval for the total causal effect that is able to capture both types of uncertainty, the uncertainty about the causal structure and the uncertainty about the numerical size of the effect. Besides the theoretical guarantees for the asymptotic coverage of the proposed confidence set, we demonstrate in the simulation section, that even in finite samples, the confidence set achieves the desired coverage probability. We emphasize again that our new approach uses the original test statistic \eqref{eq:actual_lrstat} with an upper bound on its distribution. Thus, we avoid the issues that may arise under model misspecification in the (bivariate) method by \citet{UAI21:Strieder} due to contrasting a hypothesis restricted to equal error variances with a general alternative.
 
    In the following we propose an alternative framework for constructing hypothesis tests for \eqref{eq:testproblem} based on the theory of universal inference \citep{Universal}. While this framework leads to more conservative intervals, it comes with a finite sample guarantee.

    \subsection{Split Likelihood Ratio Test}
	
	Likelihood ratio tests provide a powerful tool for constructing hypothesis tests that are calibrated via asymptotic distribution theory.  However, in the present context the needed distribution-theoretic insights are difficult to obtain and we adopted stochastic upper bounds.  An interesting alternative is to conduct the needed tests in the framework of universal inference that was introduced by \citet{Universal}. Their method employs a modification of the likelihood ratio statistic, called split likelihood ratio, based on a data splitting approach. Due to the independence of the split data, one can apply a universal critical value, instead of relying on asymptotic distribution theory. Type-I error control is then guaranteed by an application of Markov's inequality. In the following, we employ their methodology for constructing conservative confidence regions for total causal effects with finite sample guarantee. 
	
	To construct a split likelihood ratio test for the test problem \eqref{eq:testproblem}, we proceed as follows. First, we split the data into two subsets $D_0= \{X^{(1)}, \dots , X^{(k)} \}$ and $D_1 = \{X^{(k+1)}, \dots , X^{(n)} \}$ and define the log-likelihood functions
	\[
	 \ell_j(\Sigma):=\sum_{X^{(i)}\in D_j} \log p(X^{(i)}|\Sigma), \quad j=0,1,
	\]
    based on $D_0$ and $D_1$, respectively. Let $\tilde\Sigma_1:= \text{argmax}_{\Sigma \in \mathcal{M}} \ell_1(\Sigma)$ be the maximum likelihood estimator of $\Sigma$ restricted to LSEMs with equal error variances based on $D_1$.
	Then the split likelihood ratio test statistic is defined as
	\[
	    \Tilde{\lambda}^{(\psi)}_n:=2\Big( \ell_0(\tilde\Sigma_1)- \sup_{\Sigma \in  \mathcal{M}_{\psi}} \ell_0(\Sigma)\Big).
	\]
    The insight of \citet{Universal} is that for any fixed $\Sigma^* \in $ PD$(d)$ it holds that
    \[
    \E_{\Sigma}\bigg[\prod_{X^{(i)}\in D_0} \tfrac{p(X^{(i)}|\Sigma^*)}{p(X^{(i)}|\Sigma)}\bigg] \leq \int \prod_{X^{(i)}\in D_0}p(X^{(i)}|\Sigma^*) =1
    \]
    and, thus, Markov's inequality implies that for any $\alpha\in(0,1)$, under $H^{(\psi)}_0: \Sigma \in \mathcal{M}_{\psi}$, 
    \begin{align}\label{eq:slrtvalid}
        P_{\Sigma}\bigg(\Tilde{\lambda}^{(\psi)}_n>-2\log \alpha\bigg) &\leq \alpha \E_{\Sigma}\bigg[\prod_{X^{(i)}\in D_0} \tfrac{p(X^{(i)}|\tilde\Sigma_1)}{\sup_{\Sigma \in  \mathcal{M}_{\psi}} p(X^{(i)}|\Sigma)}\bigg] \\
        &\leq \alpha \E_{\Sigma}\Bigg[\E_{\Sigma}\bigg[\prod_{X^{(i)}\in D_0} \tfrac{p(X^{(i)}|\tilde\Sigma_1)}{\sup_{\Sigma \in  \mathcal{M}_{\psi}} p(X^{(i)}|\Sigma)}\Big|D_1\bigg]\Bigg] \leq \alpha. \nonumber
    \end{align}
    Thus, the decision rule
    \[ 
        \text{ reject }  H^{(\psi)}_0 \quad \text { if } \quad \Tilde{\lambda}^{(\psi)}_n>-2\log \alpha
    \]
   constitutes a valid level $\alpha$ test. This test holds level in finite samples and without any regularity conditions. Thus, we can define the following confidence region for the total causal effect based on the split likelihood ratio test.
    
    \begin{theorem}\label{theorem:slrtinterval}
	    Let $\alpha \in (0,1)$. Then a $(1-\alpha)$-confidence set for the causal effect $\mathcal{C}(i\rightarrow j)$ is given by
        \begin{align*}
            C = \{\psi \in \R  :  \Tilde{\lambda}_n^{(\psi)}(i <_G j)\leq -2\log(\alpha) \} \cup \{0 : \Tilde{\lambda}_n^{(0)}(j <_G i)\leq -2\log(\alpha)\},
        \end{align*}
      where $\Tilde{\lambda}_n^{(\psi)}(i <_G j):=\min_{ G \in \mathcal{G}(d)\,:\, i <_G j}\Tilde{\lambda}_n^{(\psi)}(G)$ and $\Tilde{\lambda}_n^{(0)}(j <_G i):=\min_{ G \in \mathcal{G}(d)\,:\, j <_G i}\Tilde{\lambda}_n^{(0)}(G)$ and  
      \begin{equation*}
           \Tilde{\lambda}^{(\psi)}_n(G):=2\Big( \ell_0(\tilde\Sigma_1)- \sup_{\Sigma \in  \mathcal{M}_{\psi}(G)} \ell_0(\Sigma)\Big).
      \end{equation*}
    \end{theorem}
    
    \begin{remark}
        Similar to the asymptotic confidence set given by Theorem \ref{theorem:lrtinterval}, we need to consider all graphs when testing for zero effects. Thus, the confidence region may consist of a non-zero interval and a disconnected zero that reflects the remaining uncertainty about the existence of an effect.
    \end{remark}

    \begin{proof}
        This result follows directly from the test inversion approach and \eqref{eq:slrtvalid}. Note that for non-zero effects we only need to consider graphs with $i <_G j$.
    \end{proof}

    \section{Computational Details and Numerical Experiments}\label{section:computation}
        
    In this section we present computational details as well as shortcuts in our proposed algorithm that significantly improve the computation time.  We also present a simulation study in which we analyze and compare the coverage probabilities and performance of the different proposed testing procedures.

    \subsection{Algorithm and Computational Shortcuts}

    We focus on the algorithm for calculating the likelihood ratio based confidence regions introduced in Theorem \ref{theorem:lrtinterval}, denoted with \texttt{LRT}. By adjusting the critical values in the algorithm and incorporating the data splitting procedure, the split likelihood ratio based confidence regions can be computed analogously.
    As the number of possible underlying causal structures, that is, the number of DAGs, grows superexponentially with the number of nodes, we need to carefully use combinatorial shortcuts in order to be able to compute our proposed confidence intervals for the total causal effect. 
    
    Our algorithm, presented as Algorithm \ref{alg:overview}, consists of two subroutines. The first routine \texttt{possible.order} (Algorithm \ref{alg:order}) immediately rejects all implausible causal orderings and, thus, reduces the set of possible causal orderings in which we subsequently test for all attainable effects with the second routine \texttt{testeffect} (Algorithm \ref{alg:testeffect}).  Carefully reducing the space of plausible causal orderings drastically decreases the computation time and allows us to compute our proposed confidence intervals, which take uncertainty over all possible DAGs into account, in a reasonable time for a moderate (but far from trivial) number of involved variables. Precise computation times can be found in Subsection \ref{section:simulation}, which presents the results of a simulation study.
    
    \begin{algorithm}[t]
    \renewcommand\figurename{Algorithm}
    \captionof{figure}{returns \texttt{LRT} confidence region  for $\mathcal{C}(i\rightarrow j)$}
    \label{alg:overview}
    \begin{algorithmic}
    \State \textbf{Input:} data, level $\alpha$, stepsize $s$
    \State (orderings, likelihood, effect) $\gets$ \texttt{possible.order}(data, $\alpha$) \Comment{all plausible orderings}
    \State L$1 \gets$ likelihood[1] \Comment{maximum likelihood alternative}
    \State startvalues $\gets $unique(effect)
    \While{startvalues is not empty}    \Comment{start with effects from unrestricted MLEs}
        \State (leftbound, rightbound) $\gets \min$(startvalues)
        \State (left, right) $\gets$ TRUE
        \While{left}
            \State leftbound $\gets$ leftbound$ - s$
            \State left $\gets $ \texttt{test.effect}(data, leftbound, orderings, L1, $\alpha$) \Comment{test for size leftbound effect}
        \EndWhile
        \While{right}
            \State rightbound $\gets$ rightbound$ + s$
            \State right $\gets $ \texttt{test.effect}(data, rightbound, orderings, L1, $\alpha$)
        \EndWhile
        \State \textbf{append} (leftbound$+\tfrac{s}{2}$,rightbound$-\tfrac{s}{2}$) to intervals \Comment{possibly multiple intervals}
        \State \textbf{remove} all values smaller than rightbound from startvalues
    \EndWhile
    \State zero $\gets$ \texttt{test.effect}(data, 0, orderings, L1, $\alpha$) \Comment{test for size zero effect}
    \State \textbf{return} (zero, intervals)
    \end{algorithmic}
    \end{algorithm}

    In order to reduce the space of plausible graphical structures as much as possible, we note again that we do not need to consider all possible DAGs but merely complete ones, represented by all possible causal orderings. This immensely decreases the worst-case number of single hypotheses we need to perform to reject an effect of size $\psi$ to $d$ factorial. Our subroutine \texttt{possible.order} (Algorithm \ref{alg:order}) further reduces the number of plausible causal orderings as follows.
    
    First, we quickly approximate the maximum likelihood under the alternative with the maximum likelihood $\widehat{\mathrm{L}1}$ obtained under the ordering obtained from the method of \cite{ChenEV} that sorts conditional variances. Subsequently, we calculate the maximum likelihood of all possible orderings without restrictions on the causal effect and immediately reject implausible orderings by employing the (largest) critical value $\chi_{d,1-\alpha}^2$. Note that we can search through the space of orderings recursively starting from the first node and immediately reject (partial) orderings if the corresponding (partial) maximum likelihood already exceeds the threshold. Further, for the subsequent testing procedure, the specific ordering of the parent set $p(i)$ is not relevant. Thus, for all orderings that only differ in the specific ordering of $p(i)$, we subsequently just test with the parent order that achieves the maximum likelihood. Additionally, we only need to test for a zero effect in the maximum likelihood ordering out of all orderings with $j<_\mathcal{G}i$. Finally, we include an attainable effect in our confidence set if it can not be rejected for at least one of the possible causal orderings using the subroutine \texttt{testeffect}, Algorithm \ref{alg:testeffect}. Therefore, we do not need to continue testing for a specific value of the effect if we find an ordering such that we cannot reject the value. Thus, first sorting the causal orderings by their maximum likelihood further improves the computation time.

    \begin{algorithm}[t]
    \renewcommand\figurename{Algorithm}
    \captionof{figure}{\texttt{possible.order}() returns plausible causal orderings}
    \label{alg:order}
    \begin{algorithmic}
    \State \textbf{Input:} data, level $\alpha$
    \State crit.val $\gets $ qchisq($1-\alpha,d)$  \Comment{$1-\alpha$ quantile of $\chi^2_d$-distr.}
    \State $\hat{\Sigma} \gets $ cov(data)
    \For{$1:d$}
    \State \textbf{append} which.min($\hat{\Sigma}_{\cdot\cdot|\mathrm{var.order}}$) to var.order \Comment{ordering of cond. variances}
    \EndFor
    \State $\widehat{\mathrm{L}1} \gets$ \texttt{get.likelihood}(var.order) \Comment{(unrestricted) maximum likelihood}
    \For{$1:d$}
    \For{order in orderings}
    \For{$v$ in setdiff($1:d$, order)}
    \State part.order $\gets$ c(order,$v$)
    \State likelihood $\gets$ \texttt{get.likelihood}(part.order) \Comment{(partial) maximum likelihood} 
    \If{2($\widehat{\mathrm{L}1}-$likelihood) $\leq$ crit.val}
    \State \textbf{append} part.order to orderings
    \If{$v$ equals $i$}
    \State effect $\gets$ $\hat{\Sigma}_{i,j|\mathrm{order}}/\hat{\Sigma}_{i,i|\mathrm{order}}$ \Comment{causal effect of maximum likelihood}
    \EndIf
    \EndIf
    \EndFor
    \State \textbf{remove} order from orderings
    \EndFor \Comment{order only relevant after $i$}
    \State \textbf{choose} order with max likelihood out of all orders in orderings that do not contain $i$ \State and are defined on the same set of nodes 
    \EndFor \Comment{keep only max likelihood for zero effect}
    \State \textbf{choose} order with max likelihood out of all orders in orderings with $j<_\mathcal{G}i$ 
    \State \textbf{sort} orderings for decreasing likelihood 
    \State \textbf{return} (orderings, likelihood, effect)
    \end{algorithmic}
    \end{algorithm}
    
    We start testing within the plausible causal orderings with the minimal causal effect that corresponds to the maximum likelihood estimates. Then, we step-wise decrease the value until we reject in all orderings. The last value we did not reject is a lower bound for our confidence region. We repeat the procedure with increasing steps until we reject in all orderings and the last value we did not reject is an upper bound for this part of our confidence region. If there remain causal effects from the maximum likelihood estimates that exceed this upper bound, we repeat this procedure. Thus, using this strategy, we might obtain a confidence region that consists of multiple intervals. Finally, we test whether we have to include zero in our confidence interval, that is, whether there remains uncertainty about the existence of an effect.

    \begin{algorithm}[t]
    \renewcommand\figurename{Algorithm}
    \captionof{figure}{\texttt{test.effect}() tests causal effect $\mathcal{C}(i\rightarrow j)=\psi$}
    \label{alg:testeffect}
    \begin{algorithmic}
    \State \textbf{Input:} data, causal effect $\psi$, orderings, likelihood alternative L1, level $\alpha$
    \For{order in orderings}
        \If{$j<_\mathcal{G}i$}
            \State crit.val $\gets $ qchisq($1-\alpha,d-1)$
            \If{$\psi$ equals zero}
                \State likelihood $\gets$ \texttt{get.likelihood}(order) \Comment{(unrestricted) maximum likelihood} 
            \Else{}
                \State likelihood $\gets -\infty$ \Comment{reject} 
            \EndIf
        \Else{}
        \State crit.val $\gets $ qchisq($1-\alpha,d)$
        \State likelihood $\gets$ \texttt{get.likelihood}(order,$\psi$) \Comment{(restricted) maximum likelihood} 
        \EndIf
        \If{2(L1$-$likelihood)$\leq$ crit.val}
            \State \textbf{return} TRUE
        \EndIf{}
    \EndFor
    \State \textbf{return} FALSE
    \end{algorithmic}
    \end{algorithm} 
   
    \subsection{Maximum Likelihood Estimation}
	
    Implementing the proposed algorithm involves the routine \texttt{get.likelihood}, which maximizes the centered Gaussian (log-)likelihood function $\ell_n(\Sigma)$ for a given causal ordering.  The function is given by
	\[
	    \tfrac{2}{n}\ell_n(\Sigma)=-\log \det (2\pi \Sigma) - \mathrm{tr}(\Sigma^{-1}\Hat{\Sigma}),
	\]
	where $\hat \Sigma = \frac{1}{n} \sum_{l=1}^n X^{(l)} (X^{(l)})^T$ is the empirical covariance. Due to the underlying assumption of homoscedasticity across all interacting variables, the set of possible distributions can be uniquely identified via \eqref{models:G}. Therefore, equivalently we maximize
	\begin{equation*}
	    \tfrac{2}{n}\ell_n(B,\sigma^2)=- d \log (2\pi \sigma^{2}) - \tfrac{1}{\sigma^2}\text{tr}((I_d-B)^{T}(I_d-B)\hat{\Sigma}),
	\end{equation*}
    using the edge matrix $B \in \R^{G}$ and the  common variance $\sigma^2>0$. The zero entries of the matrix $B$ are implied by the given causal ordering, and further, due to the acyclic structure, $B$ is permutation similar to a lower triangular matrix. Straightforward calculations yield that the maximum of $\ell_n(B,\sigma^2)$ over the equal variance $\sigma^2 >0$ is given by
    \begin{equation}\label{eq:maxprob}
	    \sup_{\sigma^2>0} \ell_n(B,\sigma^2)=- \tfrac{nd}{2} \log ( \tfrac{2\pi}{d}\text{tr}((I_d-B)^{T}(I_d-B)\hat{\Sigma})) - \tfrac{nd}{2}.
    \end{equation}
    Maximizing equation \eqref{eq:maxprob} over $B \in \R^G$ is then equivalent to minimizing
    \begin{equation}\label{eq:fulloptim}
        \mathrm{tr}((I_d-B)^{T}(I_d-B)\hat{\Sigma}))= \frac{1}{n}\sum_{k=1}^d \sum_{l=1}^n \big(X_k^{(l)}-\sum_{p \in p(k)}\beta_{k,p}X_p^{(l)}\big)^2.
    \end{equation}
    Thus, in the unrestricted case without constraints for a fixed size causal effect, the problem is equivalent to solving $d$ separate linear least squares problems with minimizer 
    \begin{equation}\label{eq:regression}
	    \hat{\beta}_{k,p(k)}^T=(\hat{\Sigma}_{p(k),p(k)})^{-1}\hat{\Sigma}_{p(k),k} \quad \mathrm{for} \quad k = 1, \dots , d,
    \end{equation}
    and corresponding minimum
    \[
        \frac{1}{n}\sum_{k=1}^d \sum_{l=1}^n \big(X_k^{(l)}-\sum_{p \in p(k)}\hat{\beta}_{k,p}X_p^{(l)}\big)^2 = \sum_{k=1}^d \hat{\Sigma}_{k,k|p(k)}.
    \]
    Further, in this case, we can calculate the maximum likelihood of any partial ordering that starts from the source node, since it collapses into subproblems. Thus, we can reject all orderings that start with a given partial ordering if the corresponding partial maximum likelihood already exceeds the threshold.
    
    In addition, we need to calculate the maximum likelihood estimates under each single hypothesis $H_0^{(\psi)}(G)$. For $j <_G i$ we only test the hypothesis $H_0^{(0)}(G)$. However, in that case, the constraint of a size zero causal effect $\mathcal{C}(i \rightarrow j)=0$ is not restrictive and, thus, the minimizer is similarly  given by \eqref{eq:regression}. Therefore, we assume in the following $i <_G j$. Analogously to the previous calculations, maximizing the Gaussian likelihood for a given causal ordering and a fixed total causal effect can be solved by minimizing the linear least squares problem \eqref{eq:fulloptim}. However, the constraint of a fixed total causal effect $\mathcal{C}(i \rightarrow j)=\psi$ links a subset of the regression parameters and we can not solve all $d$ linear least squares problems separately. 
    
    Recalling the path properties of the total causal effect, see Remark \ref{remark:paths}, the assumption $\mathcal{C}(i \rightarrow j)=(I_d-B)^{-1}_{j,i}=\psi$ restricts all $\beta_{k,p(k)}$ for $k \in p(j)\cup \{j\}\setminus p(i)\cup \{i\}$. This follows by observing that
    \[
    (I_d-B)^{-1}_{j,i}=\sum_{k=1}^{d-1}(B^k)_{j,i}
    \]
    and further inductively 
    \[(B^k)_{j,i}=\sum_{p \in p(j)\setminus  p(i) } \beta_{j,p} (B^{k-1})_{p,i}.\]
    For all other $k \notin p(j)\cup \{j\}\setminus p(i)\cup \{i\}$, the corresponding part of the optimization problem \eqref{eq:fulloptim} can be solved separately with the (unrestricted) minimizer given by \eqref{eq:regression}. However, it remains to minimize the following non-trivial joint least squares problem
    \begin{equation*}\label{eq:regproblem}
        \sum_{k \in p(j)\cup \{j\}\setminus p(i)\cup \{i\}} \sum_{l=1}^n \big(X_k^{(l)}-\sum_{p \in p(k)}\beta_{k,p}X_p^{(l)}\big)^2 \quad \text{where} \quad  \beta_{j,i}=\psi-\sum_{m=2}^{d-1}(B^m)_{j,i}.
    \end{equation*}
    Note that we can reduce the constraint to the submatrix $B_s$, which is formed by selecting all rows and columns corresponding to the set $p(j)\cup \{j\}\setminus p(i)$, since $(B^m)_{j,i}=(B_s^m)_{j,i}$. Then, the constraint yields $\beta_{j,i}$ as a function of the remaining $\beta_{k,p(k)}$ for $k \in p(j)\cup \{j\}\setminus p(i)\cup \{i\}$ and we solve the remaining least squares problem numerically with quasi-Newton methods.

	\subsection{Simulations}\label{section:simulation}
	
	In this section, we present the results of a simulation study that compares the coverage frequencies and the widths of the proposed confidence regions as well as their computing times. Our simulation experiments are designed as follows. We generate synthetic data sets based on LSEMs with Gaussian errors and equal variances associated to randomly selected directed acyclic graphs in a three-step procedure. First, we randomly select a permutation of $d$ nodes, representing the causal ordering of $(X_1, \dots, X_d)$. Second, we generate edge weights for all edges in the complete graph that corresponds to the chosen causal ordering according to a normal distribution $N(\beta,0.1)$. In that way, $\beta$ reflects the average strength of the direct effects in the graph, and thus, smaller values indicate higher uncertainty about the underlying causal structure. In the next step, we prune the complete graph, that is, we include an edge with a probability of $0.9$ for dense and $0.5$ for sparse graphs. Finally, we generate $n$ samples according to the LSEM, represented by the chosen graph, with standard normal distributed errors.  For a range of edge weights $\beta$, sample sizes $n$ and dimensions $d$, we generated multiple independent data sets and determined the confidence sets for the total causal effect of $X_1$ on $X_2$ with confidence level $\alpha=0.05$ using the different proposed methods (\texttt{LRT} given by Theorem \ref{theorem:lrtinterval} and \texttt{SLRT} given by Theorem \ref{theorem:slrtinterval}). We repeated this procedure twice, for the case of a true non-zero effect and the case of no effect. We chose the tuning parameter of the \texttt{SLRT} interval, the splitting ratio, optimal in the sense of \citet{StriederSLRT}.

    \begin{table}[t]
    \centering
    \caption{Empirical coverage of $95\%$-confidence intervals for the total causal effect of $X_1$ on $X_2$ in randomly generated $6$-dim. DAGs ($1000$ replications).}
    \resizebox{0.9\linewidth}{!}{%
    \begin{tabular}{ rr|rrr|rrr|rrr|rrr }
    \toprule
    &  & \multicolumn{ 6 }{c|}{ TRUE EFFECT }  & \multicolumn{ 6 }{|c}{ NO TRUE EFFECT } \\
     &  & \multicolumn{ 3 }{c|}{ SPARSE }  & \multicolumn{ 3 }{|c|}{ DENSE } & \multicolumn{ 3 }{|c|}{ SPARSE } & \multicolumn{ 3 }{|c}{ DENSE } \\
    method & $n$\raisebox{0.1cm}{$\setminus\beta$}  & 0.05 & 0.1 & 0.5 & 0.05  & 0.1  & 0.5 & 0.05 & 0.1 & 0.5 & 0.05 & 0.1 & 0.5 \\ 
    \midrule
    \multirow{3}{*}{LRT} & 100 & 1.00 & 1.00 & 0.99 & 0.99 & 1.00 & 1.00 & 1.00 & 1.00 & 1.00 & 1.00 & 1.00 & 1.00 \\
    & 500 & 0.99 & 0.99 & 0.99 & 0.99 & 0.99 & 1.00 & 1.00 & 1.00 & 1.00 & 1.00 & 1.00 & 1.00 \\
    & 1000 & 0.98 & 0.98 & 1.00 & 0.99 & 1.00 & 0.99 & 1.00 & 1.00 & 1.00 & 1.00 & 1.00 & 1.00 \\
    \midrule
    \multirow{3}{*}{SLRT} & 100 & 1.00 & 1.00 & 1.00 & 1.00 & 1.00 & 1.00 & 1.00 & 1.00 & 1.00 & 1.00 & 1.00 & 1.00  \\
    & 500 & 0.99 & 1.00 & 1.00 & 1.00 & 1.00 &  1.00 & 1.00 & 1.00 & 1.00 & 1.00 & 1.00 & 1.00 \\
    & 1000 & 0.99 & 0.99 & 1.00 & 0.99 & 1.00 & 1.00 & 1.00 & 1.00 & 1.00 & 1.00 & 1.00 & 1.00 \\
    \midrule
    \multirow{3}{*}{Bootstrap} & 100 & \textcolor{red}{0.70} & \textcolor{red}{0.75} & 0.97 & \textcolor{red}{0.75} & \textcolor{red}{0.80} & 0.97 & 1.00 & 1.00 & 1.00 & 1.00 & 1.00 & 1.00 \\
    & 500 & \textcolor{red}{0.72} & \textcolor{red}{0.77} & 0.97 & 
    \textcolor{red}{0.78} & \textcolor{red}{0.85} & 0.98 & 1.00 & 1.00 & 1.00 & 1.00 & 0.99 & 1.00 \\
    & 1000 & \textcolor{red}{0.76} & \textcolor{red}{0.80} & 0.95 & \textcolor{red}{0.79}& \textcolor{red}{0.88} & 0.98 & 0.99 & 0.99 & 1.00 & 1.00 & 1.00 & 1.00 \\
    \bottomrule
    \end{tabular}%
    }
    \label{tab:cover}
    \end{table}
    
    We report the observed empirical coverage probabilities for $6$-dimensional graphs in Table \ref{tab:cover}. All proposed methods achieve the desired coverage frequency of $0.95$ and seem to be conservative. This is expected for the split likelihood ratio method since it is based solely on Markov's inequality. However, it is also not surprising that the \texttt{LRT} method is conservative, considering the use of intersection union testing and the fact that we set critical values based on the relaxed alternative. For comparison, we also included classical bootstrapping results in our simulation study. The bootstrapping method employs \texttt{GDS}, an established causal discovery algorithm for LSEMs with equal variance, proposed by \citet{PetersEV}. For each resampled data set, the \texttt{GDS} method uses a greedy search algorithm to maximize the likelihood and estimate the total causal effect. The results validate the theoretical findings in \citep{drton:williams:2011, UAI21:Strieder}. Bootstrapping methods do not correctly account for uncertainty in the causal structure and do not achieve the required coverage frequency.

    \begin{figure}[t]
    \setcounter{figure}{0}
    \centering
    \includegraphics[width=0.7\linewidth]{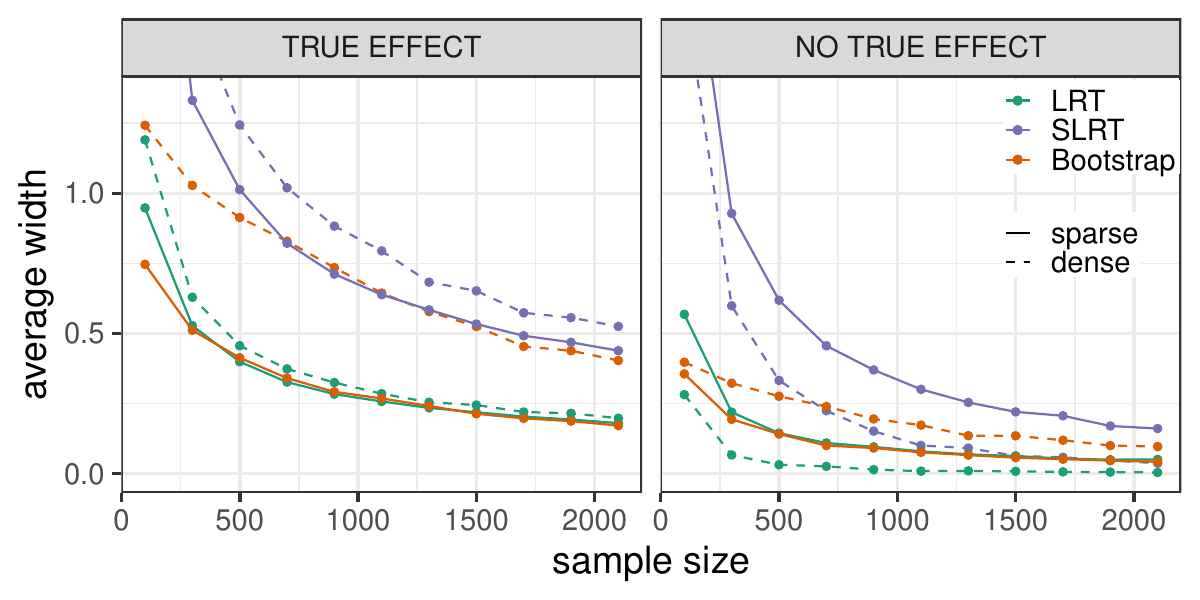}
    \caption{Mean width of $95\%$-confidence intervals for the total causal effect of $X_1$ on $X_2$ in randomly generated $6$-dim. DAGs (1000 replications).}\label{fig::width6}
    \end{figure}

    In addition to the coverage frequencies, we investigated how conclusive the confidence regions are, that is, how informative are intervals for practitioners in deciding about the existence and size of the causal effects, whilst still providing a valid uncertainty assessment.   To this we show in Figure \ref{fig::width6} the average width of the non-zero part of the confidence sets. We report the results for dimension $d=6$, where the data was generated according to random DAGs with expected edge weight $\beta=0.5$, against the sample size. In contrast to the bootstrapping method, our proposed \texttt{LRT} method constitutes a valid confidence region. Nevertheless, even in the considered setting with edge weight $\beta=0.5$, where all variants achieved the desired coverage, our method significantly outperforms bootstrapping in the dense setting, while performing similarly in the sparse setting. In contrast, the split likelihood ratio tests yield wider confidence intervals. We also note that the confidence intervals are generally wider in the dense setting, that is, the remaining uncertainty about the numerical size of the effects is higher for dense DAGs.  Our intuition here is that in dense DAGs more directed paths contribute to the total causal effect, thus, more edge weights are involved with remaining uncertainty. Figure \ref{fig::width6} shows that our proposed methods successfully help to locate the numerical size of existing causal effects while correctly quantifying remaining uncertainty.

    \begin{figure}[t]
    \centering
    \includegraphics[width=0.7\linewidth]{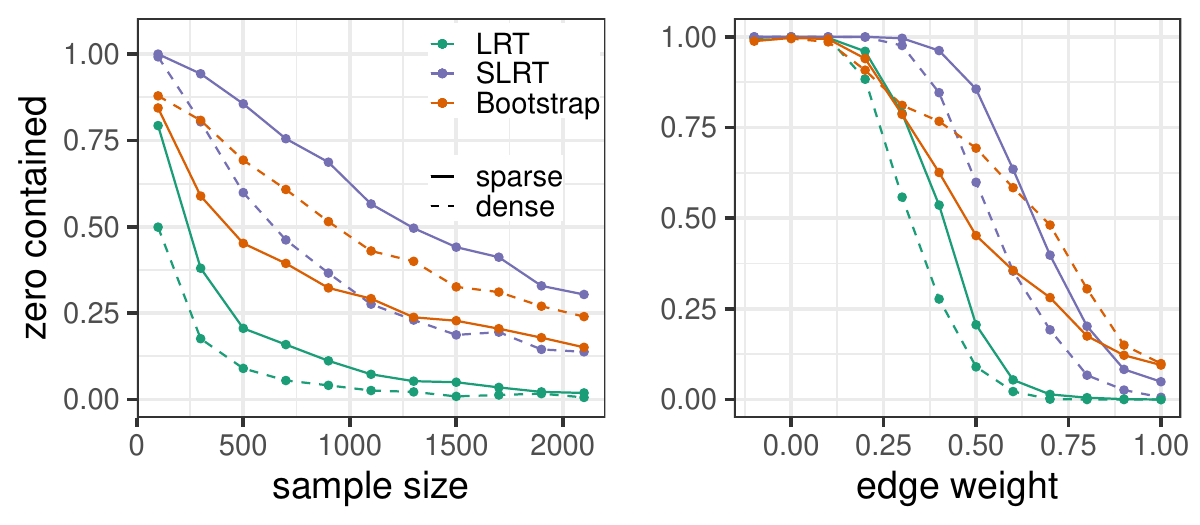}
    \caption{Percentages of times zero contained in $95\%$-confidence intervals in randomly generated $6$-dim. DAGs with non-zero effect (1000 replications). (Left) Against sample size. (Right) Against expected edge weights.} 
    \label{fig::zero}
    \end{figure}
    
    In addition to locating the numerical size of existing effects, an informative confidence interval should help in deciding whether there exists remaining uncertainty about the existence of an effect. Therefore, to investigate whether the proposed methods pick up on the causal direction of the effect between two variables, we plotted the proportions of times zero is included in the confidence sets when there is a true nonzero effect. Figure \ref{fig::zero} shows the results for dimensions $d=6$ with expected edge weight $\beta=0.5$ against the sample size and for the sample size $n=500$ against the expected edge weight. The \texttt{LRT} method eliminates the remaining uncertainty about the direction of the causal effect the quickest with increasing sample size, significantly outperforming the bootstrapping method and the split likelihood ratio tests. As intuitively expected, for higher average edge weights there is less uncertainty about the underlying causal structure and thus, the methods more often exclude the possibility of no causal effect. 
    
    Furthermore, one observes that in contrast to the bootstrapping method, the confidence regions of our proposed methods contain the zero less often in the dense setting compared to the sparse setting. This stems from the fact that there exists less uncertainty about the causal ordering for dense DAGs, and thus, the remaining uncertainty about the direction of the involved effects is lower than in the sparse setting. While there might be less uncertainty about the numerical size of the effect and thus, narrower confidence regions in the sparse setting, the uncertainty about the causal structure is higher.

    \begin{figure}[t]
    \centering
    \includegraphics[width=0.7\linewidth]{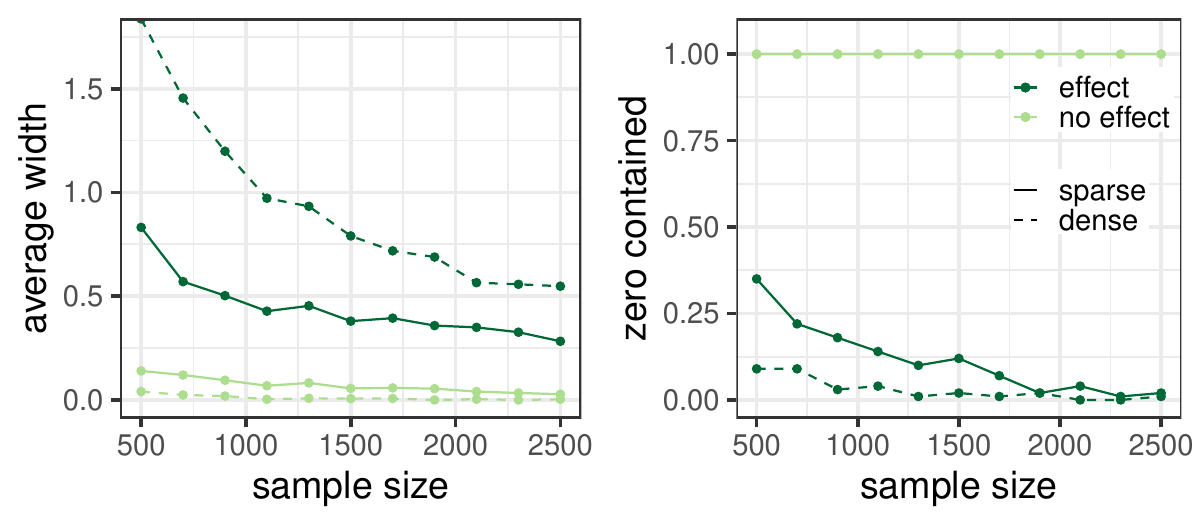}
    \caption{Mean width and percentages of times zero contained in $95\%$-confidence intervals for the total causal effect of $X_1$ on $X_2$ in randomly generated $12$-dim. DAGs (100 replications)}\label{fig::width12}
    \end{figure}

    We obtain similar results for larger causal structures with $12$ involved variables. Figure \ref{fig::width12} shows the percentages of times our proposed \texttt{LRT}-confidence regions contain zero as well as the mean width of the non-zero part. The data was generated with an expected edge weight $\beta=0.5$ and increasing sample sizes. In the case of no true effect, the confidence intervals always correctly contain zero. The average width of the non-zero part is close to zero, even for low sample sizes. If there is a non-zero total causal effect, the remaining uncertainty about the direction of the effect is already eliminated in samples of size around $2000$. We highlight again that the remaining uncertainty about the causal structure (reflected by the percentage of times zero is contained in the confidence interval) is higher in the sparse setting. In contrast, the remaining uncertainty about the numerical size of the effect, reflected by the average width of the non-zero part of the confidence intervals, is lower.
        
    In summary, the proposed framework for constructing confidence intervals for total causal effects successfully detects the direction and the numerical size of causal effects and correctly quantifies the remaining uncertainty in causal structure and effect size.
    
    \begin{figure}[t]
    \begin{subfigure}{.48\textwidth}
    \centering
    \includegraphics[width=0.95\linewidth]{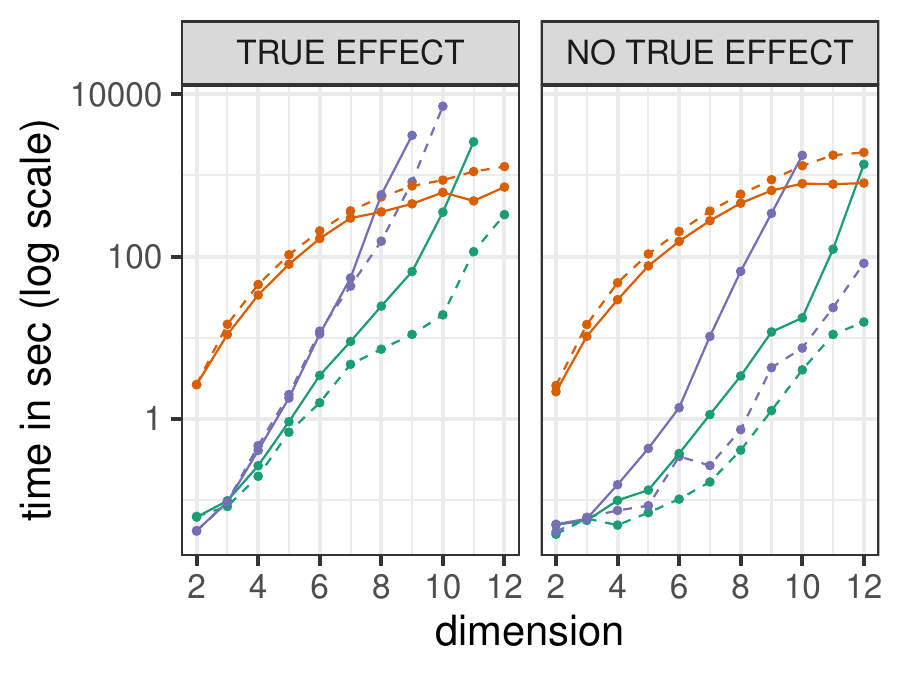}
    \end{subfigure}%
    \begin{subfigure}{.48\textwidth}
    \centering
    \includegraphics[width=0.95\linewidth]{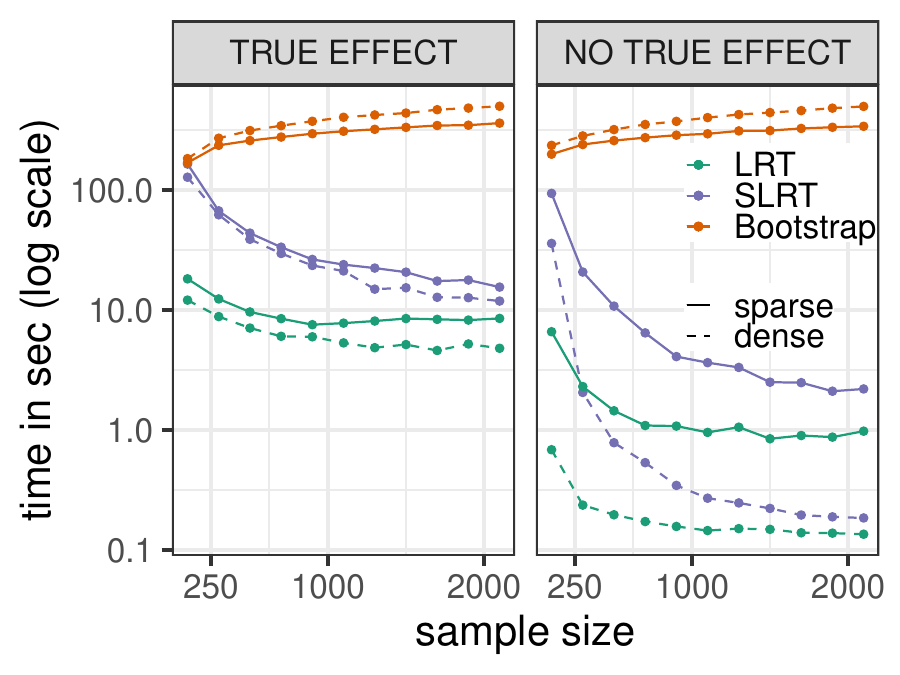}
    \end{subfigure}
    \caption{Mean computation times of $95\%$-confidence intervals for causal effect in randomly generated DAGs in seconds (100 replications). (Left) Against dimension. (Right) Against sample size.} 
    \label{fig::times}
    \end{figure}
    
    Figure \ref{fig::times} compares the average computation times over 100 randomly drawn DAGs. As the number of possible causal structures grows superexponentially with the dimension, the computation times for our proposed confidence intervals naturally increase. Since the intervals take uncertainty over all possible structures into account, the computation times similarly increase superexponentially. However, Figure \ref{fig::times} shows that it is feasible to compute this complete uncertainty over all causal structures for moderate dimensions with lower computation times than the bootstrapping method. We plotted the computation times for an expected edge weight $\beta=0.5$ and sample size $n=1000$ against the dimension (up to $d=12$). Contrasting the bootstrapping procedure based on the greedy search algorithm \texttt{GDS}, the computation times for our proposed confidence intervals are lower in the dense compared to the sparse setting. This is not unexpected, as dense DAGs generally coincide with less uncertainty in the causal structure and thus lead to fewer plausible causal orderings passed on to the testing procedure. Similarly, if there is no true effect, computation times are lower since we need to test within fewer plausible causal orderings. In general, the computation times are mainly determined by reducing the set of plausible orderings with Algorithm \ref{alg:order} and the uncertainty in the causal structure within the data.

    Further, Figure \ref{fig::times} also includes the average computation times for increasing sample sizes with expected edge weight $\beta=0.5$ and dimension $d=6$. In contrast to the bootstrapping procedure, increasing the sample size reduces the computation times for our proposed confidence intervals. More available data reduces the uncertainty about the underlying structure, which reduces the computation time.

    Our framework is based on the assumption of an underlying Gaussian LSEM with equal error variances. We emphasize that if the true data-generating mechanism follows a Gaussian LSEM with differing error variances, the causal structure and, consequently, the true causal effects are generally not identifiable. Therefore, our task of constructing a confidence interval for this causal effect is not well-defined. However, in this setting, our method targets the causal effect in an approximation of the data-generating mechanism under the equal error variances assumption. With practical applications in mind, we investigated this behavior and, thus, the sensitivity of our method towards deviations from equal error variances. We generated data according to the same procedure as described above (for $d=6$, $\beta=0.1$ and $n=500$), but with error variances sampled uniformly from $[1-0.5v,1+0.5v]$. In that way, $v$ quantifies the degree of deviation from homoscedasticity across the interacting variables, ranging from $0$ to $1.8 $. The empirical coverage probabilities in Figure \ref{fig::var} show that our framework is rather robust to small deviations from equal error variances. For small deviations, the targeted causal effect in the approximated equal error variance model is close to the true causal effect, and thus, our confidence intervals nevertheless cover the (unidentifiable) true effect.

    \begin{figure}[t]
    \centering
    \includegraphics[width=0.7\linewidth]{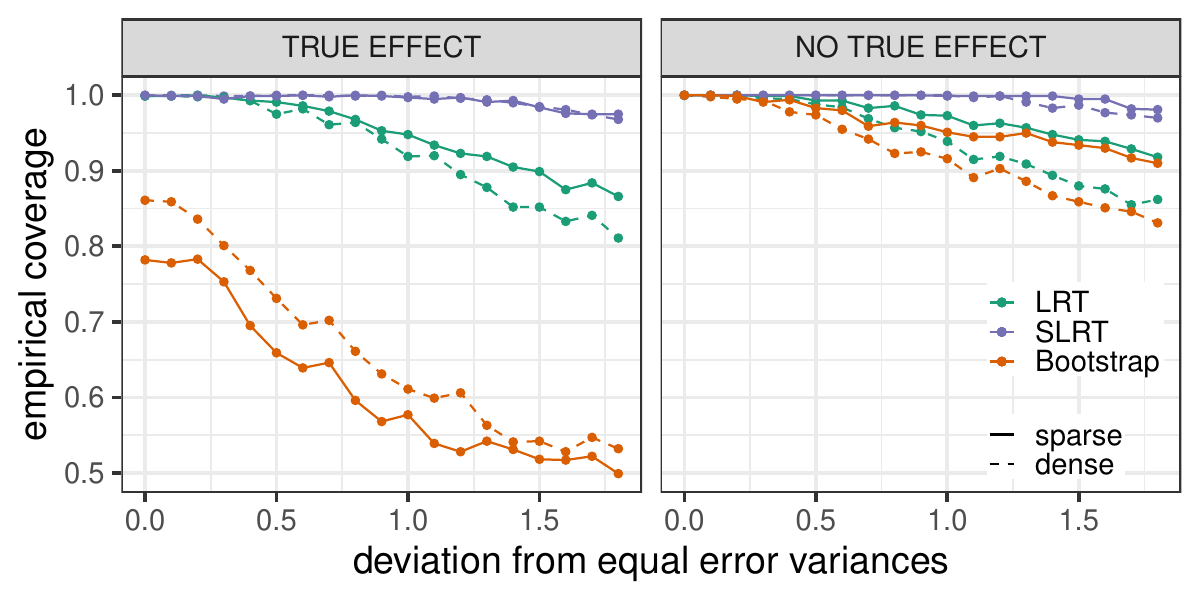}
    \caption{Empirical coverage of $95\%$-confidence intervals for the total causal effect of $X_1$ on $X_2$ in randomly generated $6$-dim. DAGs under departure from equal error variances ($1000$ replications).} 
    \label{fig::var}
    \end{figure}

    In conclusion, the likelihood ratio method performs best in our simulation studies and yields the most informative confidence intervals. Nonetheless, while the split likelihood ratio tests are more conservative, we emphasize that it is the only method with theoretical finite sample guarantees.

\subsection{Real Data Example}

    To illustrate the importance of correctly accounting for uncertainty in the causal structure, we consider the frequently studied Sachs protein data \citep{Sachs}. This collection of 14 data sets comprises expression levels for $11$ proteins and phospholipids in human T-cells obtained under different experimental conditions. In applications involving variables from similar domains and measurements obtained by similar techniques, the assumption of equal error variances offers a simple way to conduct exploratory analyses of cause-effect relations, also in light of the robustness of our method towards small deviations from equal variance that was seen in our simulations. In our data analysis, we compared our $\texttt{LRT}$ method against naive bootstrapping with the $\texttt{GDS}$ method \citep{PetersEV} as well as bootstrapping with the $\texttt{LiNGAM}$ method \citep{Lingam}. The $\texttt{LiNGAM}$ method explicitly assumes non-Gaussian errors to ensure identifiability. We considered the first of the $14$ data sets with a sample size of $853$ and focused on a subset of $10$ proteins (excluding PKA) that have measurements on a similar scale.  The protein PKA is a highly variable source node in the ground truth causal DAG reported in \cite{Sachs}.  To account for its influence, we first performed linear regressions of each protein on PKA and continued our analysis using the residuals.
    
    Figure \ref{fig::sachs} shows the calculated $95\%$-confidence intervals for $6$ exemplary causal effects. The regression coefficient of $X_i$ when regressing $X_j$ on $(X_i, X_{p(i)})$, where the conventionally accepted ground truth structure defines the parent set, is seen as the true target effect $\mathcal{C}(i \rightarrow j)$ and marked in red. In total, our proposed confidence intervals cover $84$ of the $90$ pairwise causal effects among the $10$ involved variables. In contrast, naive bootstrapping with $\texttt{GDS}$ and with $\texttt{LiNGAM}$ cover only $66$ and $72$, respectively. Our confidence intervals indicate that, even within our strict modeling assumptions, there is substantial structure uncertainty. Out of the $90$ calculated intervals, $67$ contain zero as well as non-zero effects, thus reflecting uncertainty about the direction and existence of effects. In comparison, the naive methods with $\texttt{GDS}$ and with $\texttt{LiNGAM}$ contain both directions in only $35$ and $55$ cases, respectively. Thus, while the bootstrapping methods yield narrower intervals on average, they do not correctly account for structure uncertainty and tend to miss more minor target effects.

    \begin{figure}[t]
    \begin{subfigure}{.31\textwidth}
    \centering
    \includegraphics[width=0.95\linewidth]{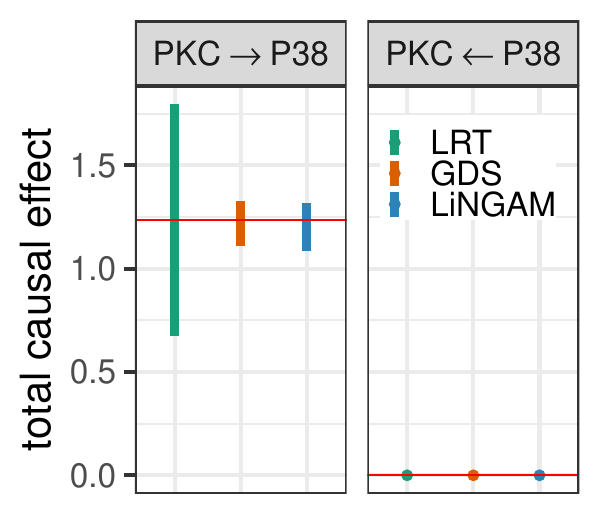}
    \end{subfigure}%
    \begin{subfigure}{.31\textwidth}
    \centering
    \includegraphics[width=0.95\linewidth]{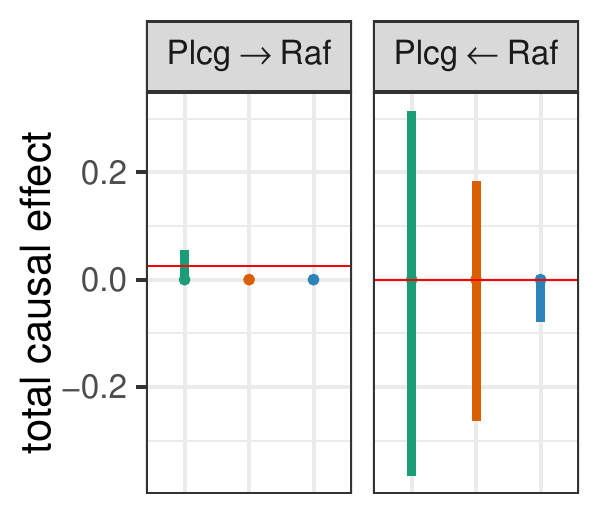}
    \end{subfigure}
    \begin{subfigure}{.31\textwidth}
    \centering
    \includegraphics[width=0.95\linewidth]{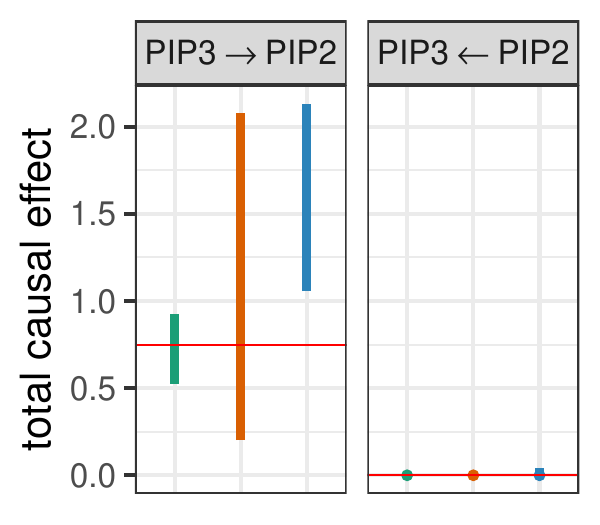}
    \end{subfigure}
    \caption{$95\%$-confidence intervals for the total causal effect between proteins from real world expression data.} 
    \label{fig::sachs}
    \end{figure}

\section{Discussion}\label{section:discussion}

In this paper, we raise the problem of rigorously quantifying uncertainty in causal inference when the causal structure is unknown. We demonstrate the intricacies of this task and propose a precise solution for constructing confidence regions for total causal effects, which capture both types of uncertainty: numerical size of effects and causal structure. By employing a test inversion approach with intersection union tests and exploiting combinatorial shortcuts, our proposed framework mathematically rigorously quantifies the remaining uncertainty in the concrete setting of LSEMs with homoscedastic Gaussian errors across all interacting variables. 

Our proposal of leveraging test inversions of joint tests for causal structure and size of effect to obtain confidence regions for causal effect estimates under structure uncertainty has a general nature and with suitable statistical tests could be extended to other modelling assumptions; see \cite{UAI21:Strieder} for results from a heuristic for the example of bivariate LSEMs with non-Gaussian errors.  Moreover, our framework could be adjusted to target different (causal) quantities by reformulating the corresponding constraint.

Assuming LSEMs with equal error variances ensures identifiability. Thus, a confidence region corresponds to a set of covariance matrices, which all uniquely entail a corresponding total causal effect. Without equal error variances (or rather identifiability in general), this unique correspondence does not hold.  While our computational framework of focusing on causal orders and complete graphs  might be adapted to consider a set of plausible covariance matrices, every covariance matrix then entails a set of possible causal effects that always contains zero. Thus, the corresponding confidence regions are not conclusive about the direction of effects. In contrast, we demonstrated that under the assumption of equal error variances our proposed confidence regions not only correctly quantify the remaining uncertainty in causal structure and numerical effect size, but also provide conclusive information which can help practitioners to decide about existence and numerical size of effects. 

The computational shortcuts we adopt allow us to quantify structure uncertainty over all possible DAGs among a moderate number of nodes. For higher dimensional systems it is not feasible to consider all possible DAGs given that there already exist more DAGs on 22 nodes ($10^{87}$) than estimated atoms in the universe. However, our method can still be used by practitioners as a valuable tool for analysing higher dimensional causal systems by partially quantifying the uncertainty in the causal structure. 
    
In the status quo scenario, practitioners would either confer with experts that provide the underlying causal structure or employ structure learning algorithms to estimate the DAG from available data. Starting from this inferred model, classical statistical methods can then be used to calculate confidence intervals for the involved causal parameters. These confidence intervals quantify uncertainty in the numerical size of the causal effect conditional on the inferred model. However, they do not incorporate uncertainty about the underlying causal structure and the (data-driven) model choice. While in higher dimensional systems it is not computationally tractable to quantify structure uncertainty by considering all DAGs, our framework can still be used as a compromise in between. Rather than solely relying on the full structure provided by experts, our method can be used to consider a set of plausible causal structures. This way practitioners can incorporate varying degrees of belief of experts in different parts of the structure by fixing parts with high confidence but still considering the remaining uncertainty over the rest. Along similar lines, instead of focusing on a single output of causal learning algorithms, practitioners can employ our method to a set of (most likely) causal structures and incorporate some degree of structure uncertainty in causal effect estimates.    

In summary, we consider our study of total causal effect in linear causal models with equal error variances as the beginning of a new line of research on confidence sets under structure uncertainty.  We anticipate fruitful generalizations in many directions including not only other model settings such as the linear non-Gaussian case, or additive noise models, but also different causal parameters other than total effects. From this perspective, while this paper develops our ideas in a very concrete setting, it sets the scene for various follow-up work in providing rigorously justified confidence regions in causal effect estimation when the causal structure itself is learned from data.

\begin{acks}[Acknowledgments]
This project has received funding from the European Research Council (ERC) under the European Union's Horizon 2020 research and innovation programme (grant agreement No.~83818), the German Federal Ministry of Education and Research and the Bavarian State Ministry for Science and the Arts. The authors of this work take full responsibility for its content.
\end{acks}

\bibliographystyle{imsart-nameyear}
\bibliography{bibliography.bib}

\end{document}